\documentclass[journal]{IEEEtran}

\usepackage{amsmath,amssymb,mathtools,bm}
\usepackage{graphicx}
\usepackage{booktabs,tabularx,array,multirow}
\usepackage{cite}
\usepackage{url}
\usepackage{xcolor}
\usepackage{tikz}
\usetikzlibrary{arrows.meta,positioning,shapes.geometric,fit,calc}
\usepackage{algorithm}
\usepackage{algorithmic}
\usepackage{comment}
\usetikzlibrary{arrows.meta,fit,positioning,backgrounds}
\usepackage{amsthm} 

\theoremstyle{boldtheorem}
\newtheorem{definition}{Definition}
\newtheorem{theorem}{Theorem}
\newtheorem{proposition}{Proposition}
\newtheorem{corollary}{Corollary}

\newcommand{\C}{\mathbb{C}}

\newcommand{\rank}{\operatorname{rank}}

\begin{document}

\title{Identifiability Limits of Forced Oscillation Sources in Power Systems}

\author{Kai~Sun,~\IEEEmembership{Fellow,~IEEE}%
\thanks{This work was supported by NSF grant ECCS-2329924.}%
\thanks{K. Sun is with the Min H. Kao Department of Electrical Engineering and Computer Science, The University of Tennessee, Knoxville, TN 37996 USA (e-mail: kaisun@utk.edu).}%
}

\maketitle

\begin{abstract}
Whether a forced-oscillation source can be uniquely localized depends
jointly on the available measurements and the candidate intervention
dictionary. This paper considers a single unknown constant-amplitude
sinusoid acting through one of physical intervention
channels. Because its amplitude and phase are unknown, each candidate
harmonic response is observable only up to a nonzero complex scalar and
therefore defines a projective ray in measurement space. A
necessary-and-sufficient condition is derived for exact localization,
together with a weighted nearest-ray locator and a deterministic recovery
bound based on projective separation. The paper proposes a framework distinguishing four
mechanisms of source localization failure: mechanism mismatch, feature-projection loss,
structural nonidentifiability, and poor conditioning or model error.
Numerical studies on the Kundur two-area system illustrate well-separated recovery,
mechanism-dependent terminal decisions, measurement-induced ambiguity,
and the sensitivity of nearly collinear signatures to model
perturbations. Then, IEEE--NASPI Contest models are used to further examine broad candidate
dictionaries, terminal-feature source coverage, and ambiguity removal
through measurement augmentation. The results show that correctly locating the source bus does not by itself establish unique identification of the internal forcing channel or robustness to measurement and model uncertainty.
\end{abstract}


\section{Introduction}
\IEEEPARstart{S}{ustained} oscillations can reduce transfer capability, complicate control-room diagnosis, and expose equipment to repeated stress. Their mitigation normally requires identifying not only the affected region but also the responsible unit and, when possible, the internal control or physical mechanism. Accordingly, oscillation source localization is naturally hierarchical: region, unit, and inside-unit mechanism. Existing approaches include traveling-wave, damping-torque, mode-shape, energy-flow, equivalent-circuit, model-based observer, and data-driven methods \cite{WangSun2017Survey,Luan2019UIO,Delabays2023PRX}.

The dissipating energy flow (DEF) method and its variants are attractive because they can use phasor measurement unit (PMU) data to localize the oscillation source at the terminal bus level, and have been successful in many power plant-dominated events \cite{Chen2013DEF,Maslennikov2017DEF,Maslennikov2024Efficiency}. Their interpretation and performance, however, depend on the choice of measured variables at terminal buses, as well as on network dissipation and the physical mechanism producing the oscillation \cite{Chevalier2018Impedance,Chevalier2020Passivity,Zhi2021EnergyFlow}. For instance, the DEF method may fail in localizing voltage--reactive power oscillations \cite{Wang2026PassivityIBR}. These findings motivate mechanism-aware improvements for oscillation source localization, but they also expose a deeper question: before selecting a physical index, \emph{is the proposed source location uniquely recoverable from the available measurements at all?}

This paper does not propose a new method for localizing the true oscillation source. Instead, it addresses the prior information-level question of whether the available measurements are sufficient to distinguish among a declared set of candidate source locations. Mechanism-based indices such as DEF and its variants may exploit useful physical structure, but they cannot overcome ambiguity already present in the raw measurements and may themselves discard distinguishing information through feature compression. Accordingly, source location is treated here as an \emph{identifiability} problem.  A single sinusoid is assumed to enter one unknown member of a finite candidate dictionary. The candidate can be an external reference, a controller summing junction, a differential or algebraic equation, a measurement or actuator channel, or a model parameter. The central observation is that the unknown amplitude and phase multiply each candidate's measured harmonic response by an arbitrary nonzero complex scalar. Consequently, a candidate is represented not by one vector but by a one-dimensional complex subspace, or projective ray. \emph{Exact localization is possible if and only if the candidate rays are distinct.} The contributions are as follows.

\begin{enumerate}
\item Source location is defined by a declared physical intervention rather than an arbitrary state coordinate, with a descriptor-system model covering interventions in dynamic and algebraic equations, signal and measurement channels, and first-order parameter modulation.

\item Projective input-distinguishability is specialized to single-frequency forced-oscillation localization with unknown amplitude and phase, yielding necessary and sufficient conditions for exact localization and an impossibility condition for observationally equivalent harmonic signatures.

\item A weighted projective harmonic-signature matching rule, bounded-error recovery condition, and sensor-selection principle that maximizes worst-case projective separation are developed.

\item A method-independent four-step framework assesses raw identifiability, robust separation, feature preservation, and mechanism-consistent interpretation. It classifies source-location failures as structural nonidentifiability, poor conditioning/model error, feature-projection loss, or mechanism mismatch, with DEF and related mechanism-based methods as illustrations.

\item Studies on the Kundur and IEEE--NASPI Contest-derived systems demonstrate well-conditioned recovery, measurement-induced ambiguity and its removal, mechanism- and feature-dependent terminal decisions, and sensitivity to finite model and measurement errors.

\end{enumerate}

In the remainder of the paper, Section~\ref{sec:formulation} defines the source and model formulation; Sections~\ref{sec:identifiability}--\ref{sec:locator} develop the identifiability theory, projective matching rule, and robustness analysis; Section~\ref{sec:failure-taxonomy} develops the four-step diagnostic framework and associated failure classes; Section~\ref{sec:numerical-studies} presents validation studies on the Kundur and IEEE--NASPI Contest systems; and Section~\ref{sec:discussion} concludes the paper.

\section{Source Definition and Problem Formulation}
\label{sec:formulation}
\subsection{Physical Intervention Sites}
In this paper, source location is defined relative to a specific physical intervention in a fixed component realization. An oscillatory input may enter through, e.g., a governor or exciter reference, measurement channel, algebraic constraint, or controller state, generally producing different responses. Each candidate source is therefore tied to a declared signal path, summing junction, equation residual, or physical parameter.

\begin{definition}[Candidate intervention]
Let $\mathcal K=\{1,\ldots,N_c\}$ denote a finite candidate dictionary.
Candidate intervention $k\in\mathcal K$ perturbs a declared physical quantity as
\begin{equation}
  q_k(t)\leftarrow q_k(t)+s(t),
  \label{eq:intervention}
\end{equation}
where
\begin{equation}
  s(t)=a\cos(\omega t+\phi), \qquad a>0.
  \label{eq:sinusoid}
\end{equation}
All other model elements remain unchanged. The \emph{source location} is the active intervention index $k$ relative to the declared dictionary and model realization. \end{definition}

Fig.~1 illustrates this hierarchy and distinguishes terminal localization from internal intervention identification. External terminal bus/branch measurements may identify a unit while leaving multiple internal candidate channels observationally equivalent.

\begin{figure}[!t]
\centering
\resizebox{0.8\columnwidth}{!}{%
\begin{tikzpicture}
[
    box/.style={
        draw,
        rounded corners,
        minimum height=6mm,
        minimum width=11mm,
        align=center,
        font=\scriptsize,
        inner sep=1.8pt
    },
    arrow/.style={-{Latex[length=1.5mm]},thick},
    inject/.style={-{Latex[length=1.5mm]},dashed,thick}
]


\node[
    draw,
    rounded corners,
    minimum width=25mm,
    minimum height=31mm,
    fill=black!5
] (regionbox) at (1.25,0) {};

\node[
    anchor=north west,
    font=\scriptsize\bfseries
] at ([xshift=1.2mm,yshift=-1.0mm]regionbox.north west)
{Region};

\node[
    draw,
    rounded corners,
    minimum width=19mm,
    minimum height=21mm,
    fill=black!0
] (unitbox) at (1.25,-0.2) {};

\node[
    anchor=north west,
    font=\scriptsize\bfseries
] at ([xshift=1.0mm,yshift=-1.0mm]unitbox.north west)
{Unit};

\node[
    draw,
    rounded corners,
    minimum width=15mm,
    minimum height=8mm,
    fill=black!15,
    align=center,
    font=\scriptsize
] (site) at (1.25,-0.5)
{Physical\\intervention\\site};

\node[
    font=\scriptsize
] (s0) at (1.25,0.45)
{$s(t)$};

\draw[arrow] (s0) -- (site.north);


\node[box,fill=black!15] (sum1) at (3.5,-1)
{Reference\\junction};

\node[box] (k1) at (4.80,-1)
{$K_1(s)$};

\node[box,fill=black!15] (sum2) at (6.1,-1)
{Internal\\junction};

\node[box] (k2) at (7.4,-1)
{$K_2(s)$};

\draw[arrow] (sum1) -- (k1);
\draw[arrow] (k1) -- (sum2);
\draw[arrow] (sum2) -- (k2);

\node[
    box,
    fill=black!15,
    minimum width=15mm
] (plant) at (7.4,0.20)
{Device and\\network port};

\node[
    box,
    minimum width=15mm
] (meas) at (7.4,1.05)
{Measured\\$y(t)$};

\draw[arrow] (k2.north) -- (plant.south);
\draw[arrow] (plant.north) -- (meas.south);


\node[
    font=\scriptsize
] (s1) at (3.5,-.2)
{$s(t)$ at $k_1$};

\node[
    font=\scriptsize
] (s2) at (6.1,-.2)
{$s(t)$ at $k_2$};

\node[
    font=\scriptsize
] (sp) at (5,0.2)
{$p^\star+s(t)$};

\draw[inject] (s1) -- (sum1.north);
\draw[inject] (s2) -- (sum2.north);
\draw[inject] (sp.east) -- (plant.west);


\node[
    draw,
    rounded corners,
    fit=(sum1)(k1)(sum2)(k2)(plant)(meas)(s1)(s2)(sp),
    inner xsep=1.2mm,
    inner ysep=3mm
] (detailbox) {};

\node[
    anchor=north west,
    font=\scriptsize\bfseries
] at ([xshift=1mm,yshift=-1.0mm]detailbox.north west)
{Expanded realization of the Unit};


\draw[dashed,thin]
    (unitbox.north east) --
    (detailbox.north west);

\draw[dashed,thin]
    (unitbox.south east) --
    (detailbox.south west);

\end{tikzpicture}
}
\vspace{-0.1in}
\vspace{-0.3in}
\caption{Hierarchical localization and physical intervention sites in a fixed component realization. A unit label need not distinguish an upstream reference intervention from an internal intervention.}
\label{fig:hierarchy}
\end{figure}
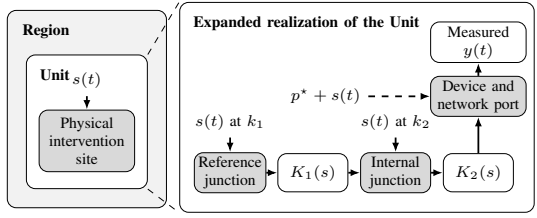

\subsection{Descriptor-System Representation for Power Systems}

Consider a nonlinear power-system DAE model
\begin{align}
    \dot{x}_d &= f(x_d,v,p,u),\\
    0 &= g(x_d,v,p,u),\\
    y &= h(x_d,v,p,u),
\end{align}
where $x_d$ contains the differential states, $v$ contains the
algebraic variables, $p$ denotes model parameters, $u$ collects
external and internal signal channels, and $y$ contains the available
measurements. Define the augmented descriptor
variable $x=[x_d^\top~ v^\top]^\top$. Linearization about an operating point ``$\star$" shifted to the origin gives, for candidate
intervention $k$,
\begin{align}
    E\dot{x}(t) &= Ax(t)+b_k s(t), \label{eq:descriptor-state}\\
    y(t) &= Cx(t)+d_k s(t), \label{eq:descriptor-output}
\end{align}
where the same symbols denote the corresponding deviations, and $b_k$ and $d_k$ define the linearized channel for candidate $k$. The direct term $d_k$ captures forcing or corruption acting directly on a measured channel.
For the semi-explicit DAE, singular $E$ retains the network algebraic variables in $v$, while $A$ and $C$ follow from the corresponding Jacobians:
\begin{equation}
    E=
    \begin{bmatrix}
        I & 0\\
        0 & 0
    \end{bmatrix},     ~A=
    \left.
    \begin{bmatrix}
        f_{x_d} & f_v\\
        g_{x_d} & g_v
    \end{bmatrix}
    \right|_\star,
    ~C=
    \left.
    \begin{bmatrix}
        h_{x_d} & h_v
    \end{bmatrix}
    \right|_\star.
\end{equation}
A time-varying physical parameter is included in the same framework.
If candidate $k$ is defined by
\begin{equation}
    p_k(t)=p_k^\star+s(t),
\end{equation}
then, to first order,
\begin{equation}
    b_k=
    \left.
    \begin{bmatrix}
        f_{p_k}\\
        g_{p_k}
    \end{bmatrix}
    \right|_\star,
    \qquad
    d_k=
    \left.h_{p_k}\right|_\star .
    \label{eq:param-channel}
\end{equation}
Thus, sinusoidal parameter modulation fits the same descriptor model in \eqref{eq:descriptor-state}--\eqref{eq:descriptor-output}, with nominal $A$ and $C$ unchanged. If $b_k=d_k=0$ at the operating point, the modulation is locally invisible at the fundamental frequency; identification then requires higher-order harmonics or sidebands, a different operating point, or measurements with a nonzero first-order signature, which are beyond the scope of this paper.

Table~I summarizes representative candidate interventions and their
corresponding linearized channels. The formulation does not assume
that all listed mechanisms are equally likely; it provides a common
mathematical representation for testing whether the available
measurements can distinguish them.

\begin{table}[t]
\caption{Representative Candidate Interventions and Linearized Channels}
\label{tab:dictionary}
\centering
\scriptsize
\setlength{\tabcolsep}{2pt}
\renewcommand{\arraystretch}{1.1}
\begin{tabular}{@{}p{0.17\columnwidth}p{0.48\columnwidth}p{0.3\columnwidth}@{}}
\toprule
Candidate & Example and linearized channel & Identifiability issue \\
\midrule

Controller/\newline reference 
& $P_{\rm ref}$, $Q_{\rm ref}$, $V_{\rm ref}$, frequency reference; known column or distribution vector at the declared summing junction
& Cascaded locations may differ only by a complex gain \\

Differential-equation residual
& Mechanical torque, excitation or controller state; inject $s(t)$ into the declared differential equation
& Location depends on the chosen realization \\

Algebraic-equation residual
& Current balance, voltage constraint, converter relation; corresponding algebraic-equation channel
& Elimination may merge distinct physical channels \\

Measurement/\newline actuator
& Sensor oscillation or modulated command; $b_k=0$, $d_k\neq0$ for direct measurement injection
& Direct and dynamic paths may be indistinguishable \\

Parameter modulation
& Gain, impedance, time constant; sensitivity channels $b_k$ and $d_k$ from \eqref{eq:param-channel}
& First-order signature may vanish at some operating points \\

\bottomrule
\end{tabular}
\end{table}

\subsection{Coordinate Dependence and Standing Assumptions}
\begin{proposition}[Coordinate dependence]
The statement that a source acts on one state coordinate is not invariant under a nonsingular state transformation. A source location is invariant only after it is tied to a physical intervention in a fixed realization.
\end{proposition}
\vspace{-0.1in}
\noindent\begin{proof}
Let $\dot x_d=f_{x_d}x_d+e_i s$, and define $z=Tx_d$ with nonsingular $T$. Then $\dot z=Tf_{x_d}T^{-1}z+Te_i s$. In general, $Te_i$ has multiple nonzero entries. Hence a one-coordinate injection in $x_d$ is not a one-coordinate injection in $z$, although the named physical intervention that generated the channel is unchanged.
\end{proof}

The principal result assumes: 1) one active candidate; 2) nonzero constant forcing amplitude; 3) a known or consistently estimated angular frequency $\omega$; 4) a known operating-point model and measurement set; 5) steady harmonic response after transients; and 6) a regular descriptor pencil with $j\omega E-A$ nonsingular. Measurement and model uncertainty are treated within this setting. The paper focuses on fundamental single-source identifiability from a linear-system perspective; finite-horizon transients, nonlinear or hybrid responses, frequency drift, and multiple simultaneous sources are outside its scope, though discussed where relevant.

\section{Fundamental Identifiability of a Harmonic Source}
\label{sec:identifiability}
\subsection{Projective Harmonic Signatures}
Write the forcing as $s(t)=\operatorname{Re}\{\alpha e^{j\omega t}\}$, where $\alpha=ae^{j\phi}\neq0$. The steady-state phasor measurement satisfies
\begin{equation}
  \widehat y=g_k(\omega)\alpha,
  \label{eq:signaturemodel}
\end{equation}
Assume that the descriptor pencil $\lambda E-A$ is regular and that
$j\omega E-A$ is nonsingular. These conditions ensure a unique harmonic
particular solution at the forcing frequency even when $E$ is singular
due to retained algebraic variables. Define
\begin{equation}
  g_k(\omega)=C(j\omega E-A)^{-1}b_k+d_k\in\C^m
  \label{eq:signature}
\end{equation}
as the \emph{candidate harmonic signature}. Because $\alpha$ is unknown, $g_k$ and $cg_k$ for any $c\in\C\setminus\{0\}$ represent the same observable source hypothesis.

\begin{definition}[Projective signature]
The amplitude- and phase-invariant signature of candidate $k$ is the complex ray
\begin{equation}
  [g_k]=\{cg_k:c\in\C\setminus\{0\}\},
  \label{eq:ray}
\end{equation}
which is a point in complex projective space $\mathbb{CP}^{m-1}$ when $g_k\neq0$, and, equivalently,  has real dimension $2(m-1)$.
\end{definition}

\begin{definition}[Observational equivalence]
Candidates $k$ and $\ell$ are equivalent at frequency $\omega$, denoted $k\sim_\omega\ell$, if $[g_k(\omega)]=[g_\ell(\omega)]$.
\end{definition}
This relation partitions the candidate dictionary into equivalence classes; the finest data-supported output is the class containing the active candidate. Fig.~\ref{fig:rays} contrasts distinguishable projective rays with an observationally equivalent pair.

\begin{figure}[!htbp]
\vspace{-0.1in}
\centering
\resizebox{0.8\columnwidth}{!}{%
\begin{tikzpicture}[scale=0.82, every node/.style={font=\scriptsize}, ray/.style={-{Latex[length=1.5mm]}, thick}]
\draw[->] (-0.2,0) -- (3.2,0) node[right] {$\operatorname{Re}$};
\draw[->] (0,-0.2) -- (0,2.7) node[above] {$\operatorname{Im}$};
\draw[ray] (0,0) -- (2.6,0.8) node[right] {$[g_k]$};
\draw[ray] (0,0) -- (1.2,2.2) node[above] {$[g_\ell]$};
\node at (1.65,-0.45) {distinct rays: identifiable};
\begin{scope}[xshift=4.6cm]
\draw[->] (-0.2,0) -- (3.2,0) node[right] {$\operatorname{Re}$};
\draw[->] (0,-0.2) -- (0,2.7) node[above] {$\operatorname{Im}$};
\draw[ray] (0,0) -- (2.7,1.35) node[right] {$[g_k]=[g_\ell]$};
\draw[ray, dashed] (0,0) -- (1.7,0.85);
\node at (1.65,-0.45) {same ray: indistinguishable};
\end{scope}
\end{tikzpicture}
}
\vspace{-0.1in}
\caption{Projective geometry induced by unknown forcing amplitude and phase. The drawing is schematic; each ray in the actual problem lies in $\C^m$.}
\vspace{-0.1in}
\label{fig:rays}
\end{figure}
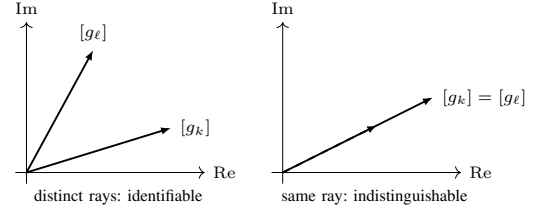

\subsection{Necessary and Sufficient Condition}
\begin{theorem}[Exact single-frequency identifiability]
All candidates in $\mathcal K$ are uniquely identifiable from the noise-free steady-state phasor $\widehat y$ for every nonzero forcing amplitude and phase if and only if
\begin{align}
  g_k(\omega)&\neq0, &&\forall k\in\mathcal K, \label{eq:nonzero}\\
  \rank_{\C}\begin{bmatrix}g_k(\omega)&g_\ell(\omega)\end{bmatrix}&=2,
  &&\forall k\neq\ell. \label{eq:pairrank}
\end{align}
\end{theorem}
\noindent\begin{proof}
If candidate $k$ is active, then $\widehat y=\alpha g_k\in[g_k]$. Under (\ref{eq:nonzero})--(\ref{eq:pairrank}), the rays are nonzero and pairwise distinct, so exactly one candidate ray contains $\widehat y$. Conversely, if $g_\ell=cg_k$ for some $c\neq0$, candidate $k$ with coefficient $\alpha$ and candidate $\ell$ with coefficient $\alpha/c$ both produce $\widehat y=\alpha g_k$. The two source locations are indistinguishable from the phasor measurements after accounting for the unknown forcing amplitude and phase. No deterministic localization rule based solely on these measurements can correctly distinguish both source locations. If $g_k=0$, candidate $k$ is unobservable and cannot be localized.
\end{proof}

The theorem is both a recovery and no-go result. It does not say that arbitrary internal forcing is always locatable. It specifies exactly when a point location is logically supported.

\begin{proposition}[Descriptor rank test]
When $j\omega E-A$ is nonsingular, condition (\ref{eq:pairrank}) is equivalent to
\begin{equation}
\rank
\begin{bmatrix}
 j\omega E-A & -b_k & b_\ell\\
 C & d_k & -d_\ell
\end{bmatrix}=n+2
\label{eq:rosenbrock}
\end{equation}
for every $k\neq\ell$.
\end{proposition}
\noindent\begin{proof}
Denote the matrix by \( M_{k\ell}\) Since $j\omega E-A$ is nonsingular, any nonzero null vector $[x^\top,\beta,\gamma]^\top$ of $M_{k\ell}$ must have $(\beta,\gamma)\neq(0,0)$.  The first block row gives \[ x=(j\omega E-A)^{-1}(b_k\beta-b_\ell\gamma). \] Substitution into the second block row yields \[ \beta\!\left[C(j\omega E-A)^{-1}b_k+d_k\right] -\gamma\!\left[C(j\omega E-A)^{-1}b_\ell+d_\ell\right]=0, \] or equivalently $\beta g_k-\gamma g_\ell=0$.  Thus, $\rank M_{k\ell}<n+2$ iff $[g_k\;g_\ell]$ has a nonzero null vector, i.e., iff $\rank_{\mathbb C}[g_k\;g_\ell]<2$. 
Therefore,
\(
\operatorname{rank} M_{k\ell}=n+2
\) iff
\(
\operatorname{rank}_{\mathbb C}[\,g_k\;g_\ell\,]=2,
\)
which proves \eqref{eq:rosenbrock}.\end{proof}

This Rosenbrock-type rank condition is closely related to input-distinguishability results in unknown-input observer and fault-isolation theory \cite{Kudva1980UIO,Hou1992UIO,Darouach1994UIO,Luan2019UIO}. The contribution is its specialization to forced-oscillation source localization: physical intervention hypotheses are defined in a power-system DAE model, unknown sinusoidal amplitude and phase are handled through projective equivalence, and region-, unit-, and internal-mechanism identifiability are distinguished. This projective formulation also serves as the information-level reference for evaluating subsequent feature maps and source-location criteria.

\begin{corollary}[One scalar phasor is insufficient]
If $m=1$, two or more candidates with nonzero scalar signatures cannot be uniquely localized when their forcing amplitudes and phases are arbitrary.
\label{scalar-phasor}
\end{corollary}
\vspace{-0.1in}

\noindent\begin{proof}
Every nonzero scalar $g_k\in\C$ spans the same space $\C$, so all nonzero candidate rays coincide.
\end{proof}
This does not imply that one physical sensor is always useless, since
time-series data may contain transient or multi-frequency information.
It states only that one steady single-frequency complex scalar,
as obtained from a sustained constant-amplitude oscillation, contains
no spatial direction for distinguishing arbitrary unknown complex gains.

\vspace{-0.1in}
\subsection{Hierarchical Identifiability}
Let $\pi_r(k)$, $\pi_u(k)$, and $\pi_i(k)$ map an intervention to its region, unit, and internal mechanism labels. At hierarchy level $q\in\{r,u,i\}$, two labels $a$ and $b$ are distinguishable if no ray generated by a candidate with $\pi_q(k)=a$ equals a ray generated by a candidate with $\pi_q(\ell)=b$. Thus, region localization can be exact even if several units in that region are equivalent; likewise, a unit can be exact while its internal candidates remain indistinguishable from the specified steady-state harmonic measurements under unknown forcing amplitude and phase. This formalizes the practical distinction between finding an oscillating area, a contributing resource, and the root control variable.

\section{Projective-Signature Matching and Robustness}
\label{sec:locator}
\subsection{Candidate Residuals}
Let positive-definite matrix $W\succ0$ weight measurement confidence, scaling, and covariance, with $\|v\|_W^2=v^HWv$. 
For candidate $k$, define the weighted residual over the unknown
complex forcing coefficient $\alpha$ as
\begin{equation}
 J_k=\min_{\alpha\in\C}\|\widehat y-g_k\alpha\|_W^2 .
 \label{eq:residual}
\end{equation}
The minimizing weighted least-squares estimate of $\alpha$ is
\begin{equation}
 \widehat\alpha_k=
 \frac{g_k^H W\widehat y}{g_k^H Wg_k}.
 \label{eq:alphahat}
\end{equation}
Substituting \eqref{eq:alphahat} into \eqref{eq:residual} gives
\begin{equation}
 J_k=
 \widehat y^H W\widehat y-
 \frac{|g_k^H W\widehat y|^2}{g_k^H Wg_k}.
 \label{eq:residual-closed}
\end{equation}
The point estimate and a thresholded set-valued estimate are
\begin{align}
 \widehat k=\arg\min_{k\in\mathcal K}J_k, \qquad
 \widehat{\mathcal K}_\tau=\{k:J_k\leq J_{\min}+\tau_k\}. 
 \label{eq:point-set-estimate}
\end{align}
The localization algorithm is given below. In the exact, noise-free case, $J_k=0$ for every candidate in the true equivalence class and is positive for every candidate outside it.

\begin{algorithm}[t]
\footnotesize
\caption{Projective Harmonic-Signature Localization}
\label{alg:locator}
\begin{algorithmic}[1]
\STATE \textbf{Offline:} define the intervention dictionary $\mathcal K$.
\STATE Construct $(b_k,d_k)$ and compute $g_k(\omega)$ for all $k$.
\STATE Form equivalence classes and the separation matrix $[\delta_{k\ell}]$.
\STATE Flag zero, collinear, and nearly collinear signatures.
\STATE \textbf{Online:} estimate the dominant forcing frequency $\widehat\omega$
and extract the corresponding harmonic measurement phasor $\widehat y$.
\STATE Evaluate $\widehat\alpha_k$ and $J_k$ for each candidate.
\STATE Return $\widehat k$ only when the residual margin and model-validity tests pass; otherwise return $\widehat{\mathcal K}_\tau$.
\end{algorithmic}
\end{algorithm}

\vspace{-0.1in}
\subsection{Separation and Deterministic Robustness}
\label{sec:robust-separation}
Define the weighted projective coherence and separation
\begin{align}
 \mu_{k\ell}^{W}=\frac{|g_k^HWg_\ell|}
 {\sqrt{(g_k^HWg_k)(g_\ell^HWg_\ell)}}, ~
 \delta_{k\ell}^{W}=\sqrt{1-(\mu_{k\ell}^{W})^2}. \label{eq:coherence-separation}
\end{align}
Here $g_k\neq0,~g_l\neq0$, and $0\leq\mu_{k\ell}^{W}\leq1$. All geometric quantities are defined in the raw measurement space unless otherwise stated.
The pairwise separation $\delta_{k\ell}^{W}$ satisfies:
$\delta_{k\ell}^{W}=0$ indicates exact projective ambiguity,
$0<\delta_{k\ell}^{W}\ll1$ indicates near ambiguity or poor conditioning,
and $\delta_{k\ell}^{W}\approx1$ indicates strong pairwise separation.
A large value for one pair does not establish identifiability of the complete dictionary. Define its global minimum separation as
\begin{equation}
 \delta_{\min}^{W}(\omega)=\min_{k\neq\ell}\delta_{k\ell}^{W}(\omega).
 \label{eq:minseparation}
\end{equation}
Structural identifiability requires every $g_k(\omega)\neq0$ and $\delta_{\min}^{W}(\omega)>0$. A candidate with $g_k(\omega)=0$ is instead unobservable and structurally non-identifiable, and its pairwise $\delta$ is undefined. Terms such as small or well separated are relative to the measurement configuration, candidate dictionary, and the error margin in (\ref{eq:robustbound}). No universal numerical threshold on \(\delta_{k\ell}^W\) is implied; the level required for reliable discrimination depends on measurement and model uncertainty.

\begin{proposition}[Bounded output-error recovery]
Suppose candidate $k$ is active and
\begin{equation}
 \widehat y=\alpha g_k+n,\qquad \|n\|_W\leq\varepsilon.
 \label{eq:noisemodel}
\end{equation}
Nearest-ray classification returns $k$ uniquely if, for all $\ell\neq k$,
\begin{equation}
 d_{k\rightarrow\ell}^{W} \overset{\mathrm{def}}{=}|\alpha|\sqrt{g_k^HWg_k}\,\delta_{k\ell}^{W}>2\varepsilon.
 \label{eq:robustbound}
\end{equation}
\end{proposition}
\noindent\begin{proof}
After whitening with $L^HL=W$, the weighted norm becomes Euclidean.
The noiseless point $\alpha Lg_k$ lies on the true ray $[Lg_k]$, and
its distance to the competing ray $[Lg_\ell]$ equals $d_{k\rightarrow\ell}^{W}$. Since $\|Ln\|_2\le\varepsilon$, the
observed point is at most $\varepsilon$ from the true ray and at least
this noiseless separation minus $\varepsilon$ from the competing ray.
Hence, if the noiseless separation exceeds $2\varepsilon$, the true
ray is strictly closer than every competing ray.
\end{proof}
This bound applies to measurement error and aggregate response mismatch
relative to a fixed nominal candidate dictionary. Under stochastic noise
with unbounded support, no finite-data method can guarantee deterministic
100\% point accuracy; the relevant outputs are confidence levels, error
probabilities, or asymptotic consistency.

\vspace{-0.1in}
\subsection{Sensor Placement}
For a selected measurement set $\mathcal M$, let $g_k^{\mathcal M}$
denote the restricted signature. This suggests the following
source-localization-oriented sensor-selection criterion: maximize the
worst-case projective separation while retaining adequate sensitivity,
\begin{equation}
 \max_{\mathcal M}\;\min_{\substack{k\neq\ell\\\omega\in\Omega}}
 \delta_{k\ell}^{W}(\mathcal M,\omega)
 \quad\text{s.t.}\quad
 \|g_k^{\mathcal M}(\omega)\|_W\geq\gamma_{\min}.
 \label{eq:sensorplacement}
\end{equation}
It is a design principle, while optimization of an actual PMU placement is beyond the scope of this paper.
This differs from maximizing modal observability alone. For monitoring inter-area oscillations, PMUs have been conventionally placed at locations to monitor variables with strong modal observabilities, which enable large responses in PMU measurements. However, for localization of a forced oscillation source, a large response is not sufficient if every candidate produces the same normalized spatial signature. 

\section{A Diagnostic Framework for Source Localization}
\label{sec:failure-taxonomy}

This section first introduces a taxonomy of source-location failures and then organizes them into a four-step diagnostic framework. The goal is not to rank or replace DEF or other mechanism-based methods, but to assess whether the source is identifiable from raw measurements, robustly separated under uncertainty, preserved by feature extraction, and interpreted with a mechanism consistent with the oscillation.

\vspace{-0.1in}
\subsection{Measurement and Feature Hierarchy}
For a declared candidate intervention $k$, measurement configuration $\mathcal M$, and nuisance variables $\eta$---including source amplitude and phase, initial condition, and measurement error, which affects the measurements but are not of direct interest---a source-location procedure can be written as the composition
\begin{equation}
 (k,\eta)
 \xrightarrow{\;\mathcal F_{\mathcal M}\;}
 y_k
 \xrightarrow{\;\Phi\;}
 z_k
 \xrightarrow{\;\mathcal D\;}
 \widehat{\mathcal K},
 \label{eq:informationhierarchy}
\end{equation}
where $\mathcal F_{\mathcal M}$ is the physical input--output map, $y_k$ is the available raw trajectory or harmonic measurement, $\Phi$ is a deterministic feature map, $z_k$ is the retained feature, and $\mathcal D$ is the final decision rule. At one forcing frequency, the raw harmonic layer is
\begin{equation}
 \widehat y=\alpha g_k(\omega)+n,
 \qquad
 z_k=\Phi(\widehat y),
 \label{eq:rawtofeature}
\end{equation}
where the complex scalar $\alpha$ contains the unknown forcing amplitude and phase. Theorem~1 characterizes identifiability at the $y$ layer. A feature-based method can attain that raw-data limit only if $\Phi$ preserves the distinctions needed by $\mathcal D$; it may nevertheless be advantageous because it incorporates physical structure, reduces dimension, regularizes noisy measurements, or enables a transparent terminal-level interpretation.

Oscillation source localization methods usually operate at the feature
layer in \eqref{eq:informationhierarchy}. For an oriented monitored
branch $e=(i,j)$, define the complex terminal integral
\begin{equation}
\mathcal W_e
=
\int \Delta I_{ij}^{*}\,d(\Delta V_i).
\label{eq:complex-energy}
\end{equation}
The conventional DEF and complex DEF (CDEF) implementations considered here can then be expressed using a unified branch-score definition
\begin{equation}
s_e(K)
=
\frac{\operatorname{Im}\{\mathcal W_e\}
+K\operatorname{Re}\{\mathcal W_e\}}
{\sqrt{1+K^2}}.
\label{eq:def-cdef-feature}
\end{equation}
For $K=0$, this gives the conventional DEF branch score
\cite{Chen2013DEF,Maslennikov2017DEF}. For nonzero $K$, \eqref{eq:def-cdef-feature}
gives the projected CDEF formulation. In the IEEE--NASPI studies,
$K=0.135$ is used following
\cite{Estevez2022CDEF,Maslennikov2024Efficiency} for the WECC test
system.
For either method, the branch scores are aggregated algebraically into
bus-level features as
\begin{equation}
z_b(K)
=
\sum_{e:\,\operatorname{from}(e)=b}s_e(K)
-
\sum_{e:\,\operatorname{to}(e)=b}s_e(K),
\label{eq:def-cdef-bus}
\end{equation}
with no branch-count normalization. Let
$\mathbf z(K)=[z_b(K)]_{b\in\mathcal B_{\rm obs}}$, where
$\mathcal B_{\rm obs}$ denotes the set of buses incident to at least one
monitored branch and therefore represented in the bus-score map. The
implemented source decision is
\begin{equation}
\mathcal D_K(\mathbf z(K))
=
\arg\max_{b\in\mathcal B_{\rm obs}} z_b(K).
\label{eq:def-cdef-decision}
\end{equation}
Thus, the bus with the largest signed score is selected; if all scores
are negative, this rule selects the least-negative one. Branch
orientation and score polarity are fixed throughout; no absolute-value
ranking or event-dependent sign reversal is used. Both DEF and CDEF are based on terminal voltage--current relations and therefore may fail to capture source information expressed primarily through a volt--var relation.

\vspace{-0.1in}
\subsection{Information Upper Bound for Feature-Based Localization}
For each candidate, let $\mathcal Y_k$ denote the set of raw measurements allowed by the nuisance variables, and define its feature-image set as
\begin{equation}
 \mathcal Z_k=\Phi(\mathcal Y_k)
 =\{\Phi(y):y\in\mathcal Y_k\}.
 \label{eq:featureimageset}
\end{equation}
The following proposition shows that a deterministic feature map cannot create distinguishability absent from the raw measurements; it can only preserve or reduce it.

\begin{proposition}[Information upper bound for deterministic features]
Let the same deterministic feature map $\Phi$ be applied under every source hypothesis. Then
\begin{equation}
 \mathcal Y_k\cap\mathcal Y_\ell\neq\varnothing
 \quad\Longrightarrow\quad
 \mathcal Z_k\cap\mathcal Z_\ell\neq\varnothing.
 \label{eq:featurebound}
\end{equation}
Consequently, a feature cannot make two raw-data hypotheses universally distinguishable when their admissible raw-measurement sets overlap. The converse need not hold: disjoint raw-measurement sets can have overlapping feature images because $\Phi$ may be many-to-one.
\end{proposition}
\noindent\begin{proof}
If $y\in\mathcal Y_k\cap\mathcal Y_\ell$, then the same value $\Phi(y)$ belongs to both $\mathcal Z_k$ and $\mathcal Z_\ell$. Conversely, a many-to-one map can send distinct inputs to the same feature value.
\end{proof}

For the single-frequency model, define
\begin{equation}
 \mathcal Y_k^{\omega}
 =\{\alpha g_k(\omega):\alpha\in\C\setminus\{0\}\}.
 \label{eq:projectiveresponseset}
\end{equation}
If $g_\ell=cg_k$ for some $c\neq0$, then $\mathcal Y_k^{\omega}=\mathcal Y_\ell^{\omega}$, and therefore every deterministic feature has the same attainable feature set for the two candidates. This is the feature-level consequence of the projective nonidentifiability condition in Theorem~1.

For illustration, if two interventions produce identical measured voltage/current trajectories, they necessarily yield the same DEF/CDEF quantity in (\ref{eq:def-cdef-feature}). The converse need not hold: distinct wide-area voltage/current patterns may produce the same scalar DEF value, terminal ranking, or projected CDEF score. Such equality may reflect feature-projection loss rather than identical raw responses. A richer CDEF representation may preserve distinctions lost by scalar DEF, but remains subject to (\ref{eq:featurebound}) because it uses the same measurements.

\vspace{-0.15in}
\subsection{Internal Ambiguity and Its Correctable Counterpart}
The distinction between an uncorrectable raw-data ambiguity and a correctable feature failure can be seen from two generic branches.

\emph{1) Uncorrectable branch---structural internal ambiguity:}
Consider the scalar cascade
\begin{equation}
 u_1\;\xrightarrow{K_1(s)}\;u_2\;\xrightarrow{K_2(s)}\;y,
 \label{eq:cascade}
\end{equation}
where candidate 1 injects before $K_1$ and candidate 2 injects after $K_1$. At a single frequency,
\begin{equation}
 g_1(j\omega)=K_1(j\omega)g_2(j\omega).
 \label{eq:cascadecollinear}
\end{equation}
Thus $[g_1]=[g_2]$. The unknown source amplitude and phase absorb the complex scalar $K_1(j\omega)$, so the two internal interventions generate the same set of steady downstream measurements. No alternative terminal feature can correct this ambiguity.  It can be removed only by introducing additional information, for example, by measuring the signal between the output of \(K_1\) and the downstream injection point, by exploiting onset-transient information when \(K_1(s)\) has distinguishable dynamics, or by otherwise restricting the admissible source waveform or internal realization.

\emph{2) Correctable branch---mechanism mismatch or feature-projection loss:}
Suppose instead that the raw response sets are disjoint but a selected feature collapses them:
\begin{equation}
 \mathcal Y_k^{\omega}\cap\mathcal Y_\ell^{\omega}=\varnothing,
 \qquad
 \mathcal Z_{k,0}\cap\mathcal Z_{\ell,0}\neq\varnothing,
 \label{eq:correctablebranch}
\end{equation}
where $\mathcal Z_{k,0}=\Phi_0(\mathcal Y_k^{\omega})$. The failure is correctable in principle if a physically appropriate, information-preserving feature $\Phi_1$ satisfies
\begin{equation}
 \Phi_1(\mathcal Y_k^{\omega})
 \cap
 \Phi_1(\mathcal Y_\ell^{\omega})
 =\varnothing.
 \label{eq:correctedfeature}
\end{equation}
This branch includes both an inappropriate physical pairing and an overly aggressive projection of otherwise distinguishable data.

Motivated by the relevance of volt--var relations to reactive-power-driven oscillation mechanisms \cite{Wang2026PassivityIBR}, this paper defines an experimental normalized $V$--$Q$ in-phase measure:
\begin{equation}
 \chi_{QV,i}
 =\frac{\operatorname{Re}\{\widehat V_i^*\widehat Q_i\}}
 {|\widehat V_i|\,|\widehat Q_i|}, \qquad
 |\widehat V_i|\,|\widehat Q_i|\neq0,
 \label{eq:qvfeature}
\end{equation}
which directly measures the in-phase relation between voltage and reactive-power oscillations. For volt--var mechanisms, it may expose relevant information more directly than the voltage--current pairing underlying conventional DEF. However, this $V$--$Q$ measure cannot overcome a raw-data identifiability limit; it can only correct a mechanism-dependent feature mismatch by retaining volt--var information that DEF or CDEF may not represent adequately. CDEF provides a related example in which a richer or rotated feature representation can recover distinctions obscured by a particular projection, as illustrated in the case studies.

However, such a $Q$--$V$ measure is not necessarily effective for mechanical-power, angle, current-control, or other mechanisms. It may also be poorly conditioned when $V$ or $Q$ has weak oscillatory visibility and cannot resolve the internal ambiguity in (\ref{eq:cascadecollinear}). It should therefore be viewed as a mechanism-consistent correction within the raw-data identifiability limit, not as a substitute for the identifiability test.
More generally, mechanism-based features can improve localization when they preserve and correctly interpret the available information, but they cannot overcome raw-data nonidentifiability. The proposed framework thus provides an information-level benchmark for mechanism-informed feature methods rather than replacing their physical source criteria.

\vspace{-0.15in}
\subsection{Poor conditioning and model error}
Unlike structural nonidentifiability arising from exact collinearity, poor conditioning is different: the candidate rays are mathematically distinct,
\begin{equation}
 g_\ell\neq c g_k\quad\text{for every }c\in\C,
 \qquad
 0<\delta_{k\ell}^{W}\ll1,
 \label{eq:nearcollinear}
\end{equation}
so exact noise-free localization remains possible, but the distance $d_{k\rightarrow\ell}^{W}$ from the true harmonic response to a competing ray is small. 
Unique nearest-ray recovery is guaranteed by (\ref{eq:robustbound}). When $d_{k\rightarrow\ell}^{W}$ is comparable to measurement error, phasor-extraction error, or the output displacement caused by model mismatch, the nominally distinct candidates can no longer be discriminated reliably. This is practical ambiguity rather than structural equivalence. This is illustrated by a case study on the Kundur system in Section~\ref{sec:kundur-validation}.

Model error can be written explicitly as
\begin{equation}
 \widehat y=\alpha(g_k+\Delta g_k)+n,
 \qquad
 \widetilde g_j=g_j+\Delta g_j,
 \label{eq:modelerror}
\end{equation}
where $\Delta g_k$ changes the realized response and $\Delta g_j$ changes the dictionary rays used by the locator. Even with a moderate nominal $\delta_{k\ell}^{W}$, sufficiently large or directionally adverse perturbations can reduce the realized separation or reverse a residual ordering. Conversely, improved calibration, sensor placement, or uncertainty-aware set-valued decisions can restore a reliable margin without changing the structural identifiability class.

As feature-based methods, DEF and CDEF compress the terminal measurements differently. Because CDEF retains a richer representation than DEF, it may preserve distinctions between candidate sources that are lost under the DEF projection. Let $s_k^{\Phi}$ be the score assigned to candidate $k$ by a feature-based method and let $\rho_k^{\Phi}$ bound its perturbation under the selected window, model, and measurement uncertainty. A sufficient ranking margin against candidate $\ell$ is
\begin{equation}
 s_k^{\Phi}-s_\ell^{\Phi}
 >\rho_k^{\Phi}+\rho_\ell^{\Phi}.
 \label{eq:featuremargin}
\end{equation}
If this margin is small, modest changes in filtering, terminal data, resistance compensation, or feature projection can change the ranking even though the raw candidates are not exactly equivalent. A richer DEF variant may improve the feature margin for a particular mechanism, but no feature can provide a robust singleton when the available raw-measurement-space separation is below the aggregate uncertainty floor.

\begin{table*}[!t]
\caption{Four-Step Diagnostic Framework and Associated Source-Location Failure Classes.}
\label{tab:failuretaxonomy}
\centering
\setlength{\tabcolsep}{3pt}
\renewcommand{\arraystretch}{1.4}
\scriptsize
\begin{tabularx}{\textwidth}{@{}p{0.14\textwidth}p{0.27\textwidth}p{0.28\textwidth}X@{}}
\toprule
Failure class & Mathematical diagnosis & Illustrations & Resolution \\
\midrule

Mechanism mismatch
& Raw data may be distinguishable, but the chosen variable pairing, supply rate, sign convention, or interpretation is inconsistent with the active mechanism; no specific $\delta_{k\ell}^{W}$ defines this class
& Conventional DEF may be uninformative for a voltage--var mechanism, while a $Q$--$V$ or other mechanism-consistent feature uses the relevant relation
& Retain the raw measurements, use a mechanism-consistent feature, and verify preservation of the raw distinction \\

Feature-projection loss
& $\mathcal Y_k\cap\mathcal Y_\ell=\varnothing$ but $\Phi(\mathcal Y_k)\cap\Phi(\mathcal Y_\ell)\neq\varnothing$; feature extraction destroys raw distinguishability
& Distinct terminal responses may yield the same scalar DEF value or ranking, while a richer CDEF representation may preserve the distinction
& Use the raw projective signature or a richer feature map and test distinguishability over the declared candidate set \\

Structural \newline nonidentifiability
& $g_\ell=cg_k$, $c\neq0$, or $g_k=0$; for a nonzero pair, $\delta_{k\ell}^{W}=0$
& DEF, CDEF, $Q$--$V$, harmonic power, and any other deterministic feature receive observationally equivalent measurements
& Add measurements, transient information, waveform constraints, or intervention knowledge; otherwise return the equivalence class \\

Poor conditioning \newline and/or model error
& Candidate rays are distinct but close, $0<\delta_{k\ell}^{W}\ll1$, and/or perturbations are comparable to $d_{k\rightarrow\ell}^{W}$
& DEF or CDEF rankings may change under small changes in window, noise, model parameters, or projection
& Improve sensing or calibration, quantify uncertainty, and use confidence-aware or set-valued decisions when the robust margin is insufficient \\

\bottomrule
\end{tabularx}
\end{table*}

\vspace{-0.15in}
\subsection{Four-Step Diagnostic Procedure and Failure Classes}
Table~\ref{tab:failuretaxonomy} summarizes the failure class associated with each stage of the diagnostic framework. The classes are defined by where the diagnostic sequence fails, rather than by which source-location method is used. DEF, CDEF, and their extensions serve as representative feature-based methods for illustrating the last two steps. The complete diagnostic sequence therefore consists of four steps:

\begin{enumerate}

\item \textit{Raw identifiability}: determine whether candidate rays $[g_k]$ and $[g_\ell]$ are equivalent, i.e., $g_\ell=cg_k$ for some $c\neq0$, or distinct. With only one complex harmonic coefficient, any two nonzero scalar signatures are projectively indistinguishable when forcing amplitude and phase are unknown.

\item \textit{Robust separation}: if $[g_k]\neq[g_\ell]$, assess whether $d_{k\rightarrow\ell}^{W}$ is large enough relative to uncertainty. Nearly collinear signatures, $0<\delta_{k\ell}^{W}\ll1$, are identifiable in principle but may exchange residual rankings under modest measurement or model error.

\item \textit{Feature preservation}: determine whether the selected feature preserves raw-data distinctions. With $\mathcal Z_k=\Phi(\mathcal Y_k)$, feature-projection loss occurs when $\mathcal Y_k\cap\mathcal Y_\ell=\varnothing$ but $\mathcal Z_k\cap\mathcal Z_\ell\neq\varnothing$. Distinct terminal responses may, for example, yield the same scalar DEF score, while a richer CDEF representation may preserve their distinction.

\item \textit{Mechanism-consistent interpretation}: determine whether the retained feature uses variables and an interpretation appropriate to the active mechanism. A voltage--var mechanism, for example, may remain distinguishable in the retained feature yet be misclassified by a criterion that ignores the $Q$--$V$ relation.
    
\end{enumerate}

To summarize, the objective of this paper is to establish when source localization is fundamentally possible from the available harmonic measurements and to distinguish this measurement-level limit from failures introduced by poor conditioning, feature extraction, or mechanism-dependent interpretation. For one constant-amplitude sinusoidal intervention observed through steady-state harmonic measurements, exact source localization is possible if and only if all admissible candidate signatures are nonzero and pairwise projectively distinct. The paper establishes the exact identifiability condition, shows how this limit can be attained under the stated assumptions, quantifies robustness to bounded perturbations, and provides the above taxonomy for diagnosing why a practical localization procedure may nevertheless fail.

\vspace{-0.1in}
\section{Numerical Studies}
\label{sec:numerical-studies}

\subsection{Common Numerical Protocol}
The numerical studies use Powertech's TSAT on the Kundur two-area system and IEEE--NASPI Contest models to examine whether: 1) the measurements distinguish candidate sources; 2) the distinction remains robust to measurement and model uncertainty; and 3) terminal features preserve distinctions present in the raw measurements. All cases use one constant-frequency sinusoidal source, and each comparison uses a common base model, operating point, forcing frequency, and measurement set.

For each candidate, a reference simulation provides the steady-state harmonic signature used to construct the dictionary, while separate simulations provide test events. Multiple settled windows from a test trajectory may provide repeated harmonic estimates; independent trajectories and repeated windows are reported separately. Harmonic coefficients are estimated by the same least-squares procedure within each comparison.

For the model-based benchmark, each measured harmonic vector is matched to the candidate dictionary using the normalized residual $J_k/\|\widehat y\|_W^2$, with the minimum-residual candidate selected. Here $W$ provides deterministic engineering normalization, or may be chosen as a precision matrix when measurement covariance is available. Under uncertainty, (\ref{eq:point-set-estimate}) may retain multiple candidates; this denotes practical ambiguity, not an exact equivalence class unless the signatures are analytically identical up to complex scale. Empirical harmonic-estimation errors do not replace the deterministic condition in (\ref{eq:robustbound}). Parameter perturbations may also alter the candidate dictionary and are therefore used only to illustrate robustness beyond Proposition~3.

For terminal-feature comparisons, the DEF ($K=0$), CDEF ($K=0.135$), and $Q$--$V$ implementations are fixed within each comparison. The variable pairing, branch orientation, projection parameters, and bus aggregation are held fixed. These are representative rather than exhaustive implementations. A terminal-bus method can identify only buses in its score map; sources outside that set are reported as outside its native decision coverage rather than as feature-level misclassifications.

\vspace{-0.1in}
\subsection{Validation on a Two-Area System}
\label{sec:kundur-validation}

\noindent \emph{1) System setup and well-separated recovery}

The 11-bus, four-generator Kundur two-area model provided with TSAT is studied using two four-candidate dictionaries: $\mathcal K_P$, consisting of TGOV1 governor-reference interventions at G1--G4, and $\mathcal K_V$, consisting of ESAC4A voltage-reference interventions at G1--G4. Each experiment fixes the base model and intervention mechanism and identifies the active generator. The weighting matrix $W$ in \eqref{eq:alphahat} uses engineering scales of 0.01~pu for voltage/current, $\pi/180$~rad for angles, 10~MW/Mvar for power, and 0.01~Hz for speed. Measurements include bus voltages; generator voltage, power, speed, and current; and branch power and current. Cases representing the regimes in Table~II are selected from pairwise candidate-signature separations. Sinusoidal reference perturbations at 0.4 or 0.7~Hz are applied from $t=4$--5~s through a 1.5-s raised-cosine ramp and sustained thereafter. Simulations last 30--40~s, with harmonic analysis over settled integer-cycle windows beginning at $t=12$--16~s. The fitting frequency is estimated from the response and then fixed for source localization.

A well-separated 0.4-Hz governor-reference case serves as a positive control before the four failure classes in Table~\ref{tab:failuretaxonomy}. The $\mathcal K_P$ dictionary uses eight channels: branch-current magnitudes on lines 2--6, 3--11, and 5--6; the imaginary component of G3 current; active powers of G3 and G4; and branch active powers on lines 4--10 and 3--11. All signatures are nonzero, with minimum pairwise projective separation 0.952. All eight independent test events identify the correct source uniquely, with a minimum best-to-second-best normalized squared-residual margin of 0.710. This provides a well-separated reference for the subsequent ambiguity and robustness studies.

\noindent \emph{2) Terminal information and its interpretation}

The following examples show that terminal processing can reduce source separation or alter the final decision even when the raw measurements remain distinguishable.

The first experiment examines mechanism dependence. A 0.4-Hz voltage-reference oscillation is applied at G1 and compared against the full $\mathcal K_V$ dictionary. The raw harmonic signatures are well separated, with a minimum pairwise projective separation of 0.912, and the raw projective comparison correctly identifies G1. Using the terminal features, however, the implemented DEF identifies G1 in only two of four settled windows from the same event, whereas the in-phase $Q$--$V$ score identifies G1 in all four. This example shows that a terminal feature and its decision rule may fail to preserve the source distinction already present in the raw measurements, while a feature more consistent with the voltage--var mechanism can improve the decision. The result does not, however, establish a unique cause for each incorrect DEF window.

An exploratory re-analysis illustrates possible decision-stage failure in a fixed automatic DEF implementation. Two independent 0.4-Hz governor-reference trajectories are used for each of G2 and G4 under common timing settings. Although the G2 and G4 cases have well-separated raw signatures ($\delta^W=0.977$) and distinct four-bus DEF vectors, the largest-signed-slope rule assigns all four trajectories to G2 (Fig.~\ref{fig:kundur-decision}). For the G4 cases, all slopes are negative, so the rule selects the least-negative terminal. Sign- or orientation-admission logic could be intentionally designed to prevent this failure; this example illustrates decision-stage information loss.

\begin{figure}[t]
\vspace{-0.1in}
\centering
\includegraphics[width=0.9\columnwidth]{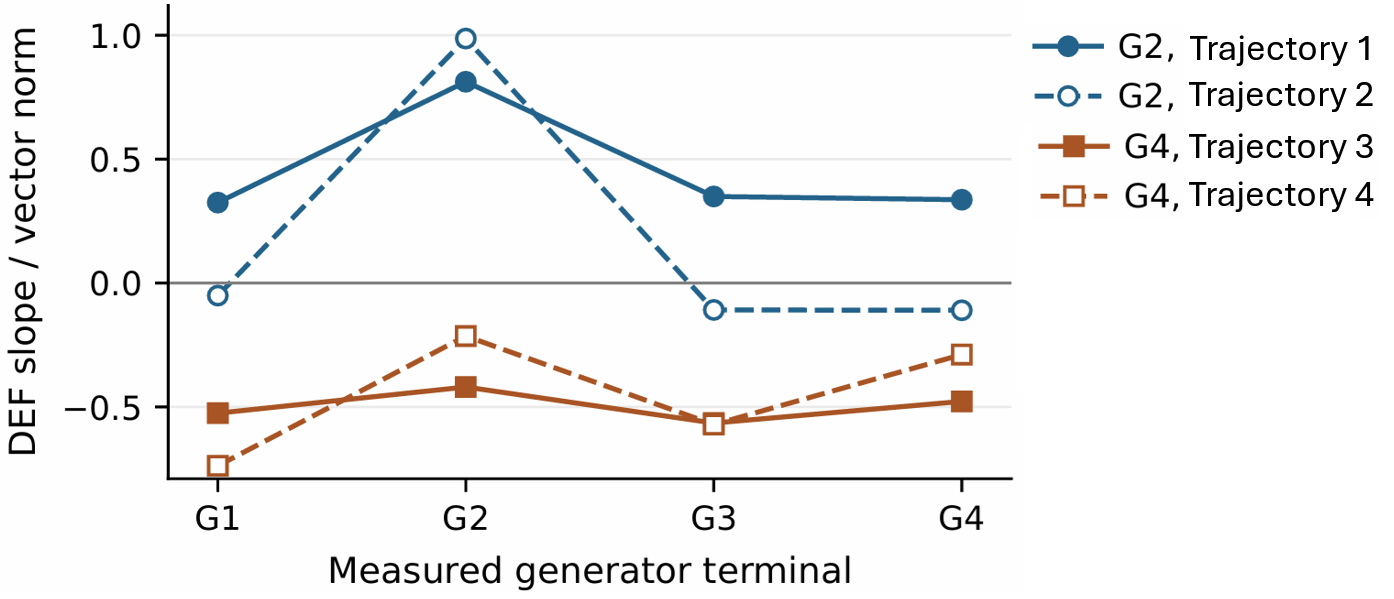}
\caption{Exploratory reanalysis of four 0.4-Hz governor-reference trajectories, two each from G2 and G4. Although the full terminal DEF vectors differ, the fixed largest-signed-slope rule selects G2 in all four cases.}
\vspace{-0.1in}
\label{fig:kundur-decision}
\end{figure}

A second experiment, using independently simulated one G2 event and one G3 event, examines the effect of aggregating the terminal DEF features into two area-level scores. The raw G2/G3 signatures are well separated, with projective separation 0.975. Two held-out 0.4-Hz governor-reference test trajectories are then evaluated, one with the source at G2 and one at G3. After area aggregation, the DEF feature separation decreases to 0.562, but pairwise matching in the reduced feature space still identifies both test trajectories correctly. Thus, the aggregation reduces source separation without causing a localization failure in this case. Because raw-data-space and feature-space distances have different scaling and uncertainty levels, their numerical distance values should not be compared against a common robustness threshold.

\noindent \emph{3) Exact ambiguity and poor conditioning}

For G3 and G1 governor-reference forcing at 0.4~Hz, retaining only the complex fundamental coefficient of active power on line 7--8 yields a single nonzero complex-scalar measurement. Therefore, the two candidate signatures lie on the same projective ray in $\mathbb C$ and are exactly indistinguishable under unknown forcing amplitude and phase. 
Adding the line-current angle on line 10--11 increases the G1/G3 separation to 0.988, illustrating how an additional measurement can remove the ambiguity.

Poor conditioning differs from exact ambiguity because the candidate rays remain distinct. At 0.7~Hz, active power and current magnitude on line 7--8 give a G3/G4 separation of only $3.08\times10^{-5}$ in $\mathcal K_P$. A separate G3 test event is correctly classified by the nominal dictionary, but with a normalized squared-residual margin of only $4.47\times10^{-10}$. Under six $\pm1\%$ perturbations of load, line reactance, or inertia, two cases are misclassified. With a residual-set tolerance of $10^{-5}$, the residual differences are too small to exclude any candidate, so the nominal event and all six perturbed cases retain G1--G4 in the admissible candidate set. This behavior reflects extreme sensitivity to a small positive separation, not exact nonidentifiability.

Collectively, these case studies on the two-area system distinguish well-separated recovery, terminal-feature and decision effects, structural nonidentifiability, and poor conditioning within the framework of Table~\ref{tab:failuretaxonomy}.

\vspace{-0.15in}
\subsection{Studies on IEEE--NASPI Contest Models}
\label{sec:naspi-validation}
\noindent \emph{1) System configuration and case designs}

The 2021 IEEE--NASPI OSL Contest provides synthetic PMU data for 13 forced-oscillation events generated by time-domain simulation of a 243-bus WECC system with 109 synchronous and 37 renewable generators \cite{MaslennikovWang2022Contest,Wang2025PMUOSL}. Cases~1, 3, 4, and 7 serve as four seed models from which TSAT generates five controlled test datasets. Within each comparison, the base model, operating point, forcing frequency, and measurement configuration are fixed.
The contest terminal layout, denoted $M_{\rm T}$, contains synchronized bus-voltage and branch-current magnitudes and angles, totaling 292 harmonic channels. The extended layout $M_{\rm T+G}$ adds generator voltage, active and reactive power, rotor angle, speed, and field voltage, for 428 channels. Measurements are normalized by 0.01~pu for per-unit quantities, $1^\circ$ for angles, and 10~MW/Mvar for power. Harmonic coefficients are estimated over three settled 6-cycle windows at the known source frequency; independent frequency estimation is used only as a consistency check.

Four contest-derived studies examine stages of the source-location procedure in \eqref{eq:informationhierarchy}, rather than providing one case per failure class in Table~II. They test whether distinctions established on the two-area system persist in a much larger model with broader candidate dictionaries and realistic terminal measurements. Case~1 demonstrates end-to-end recovery, Case~4 candidate-set/source-coverage limitations, Case~7 mechanism mismatch, and Case~3 measurement-induced structural nonidentifiability and its removal by sensor augmentation.

\noindent \emph{2) Case 1---end-to-end recovery:}

The dataset on Case~1 evaluates bus-level localization among a broad 0.82-Hz governor-reference forcing dictionary. The original Case~1 source is with a generating unit at bus 1431. Of the 109 synchronous generators, 103 in-service units (spanning 48 buses) have TGOV1, GAST, or HYGOV governor models with compatible reference channels and are therefore included as candidate interventions for testing. Multiple admissible intervention channels may exist at the same bus. Thus, for source localization at the bus level, the residual is defined as the minimum channel residual among all candidates located at that bus. The minimum cross-bus projective separation is 0.0331 for terminal set $M_{\rm T}$ and 0.0332 for $M_{\rm T+G}$. These values establish the distinguishability at the bus level, not necessarily within the unit. Using Case~1 as the seed scenario, 21 controlled test events are independently simulated, including 12 at bus 1431 and nine at selected competing or same-area buses. All 21 events are localized to the correct bus over three settled windows per event, demonstrating successful end-to-end source localization through the procedure in \eqref{eq:informationhierarchy}.

\noindent \emph{3) Case 4---candidate-set/source-label coverage:}

The original Case~4 contains a 0.379-Hz governor source at bus 3831. In this study, the Case~4 model is retained as the base system, while the original forcing is removed and 0.3-Hz governor-reference interventions are introduced separately via two other California-area units at buses 2438 and 3933. Only four regional voltage/current terminal pairs on branches 3234--3204, 3333--3303, 3933--3923, and 3903--3301 are used, providing 16 scalar measurement channels. The two raw source signatures are well separated, with a projective separation of 0.905, compared with a maximum within-source variation of only 0.0162. Using the 103-channel governor-reference dictionary, both source buses are correctly identified in all three windows of each of 12 independently simulated trajectories per source. 

The corresponding DEF responses also remain distinct. However, if the DEF implementation restricts candidate source locations to the eight terminal buses of the four directly measured branches, bus 2438 is outside the candidate set and therefore cannot be identified by the resulting DEF localization rule. Thus, unlike Case~7, where the source information is lost or misinterpreted through the feature and decision stages, Case~4 preserves source distinction through the DEF feature stage, $\Phi$,  but cannot achieve exact bus localization at the final decision stage $\mathcal D$ in \eqref{eq:informationhierarchy} because the true source bus is excluded from its candidate set.

\noindent \emph{4) Case 7---mechanism dependence:}

The Case~7 study examines mechanism dependence for a 0.379-Hz source at bus 2634. Two sets of 12 independent trajectories use the same source unit with governor- and exciter-reference forcing, respectively, while varying amplitude, phase, and onset time. The resulting raw dictionary contains eight candidate intervention channels at three buses. As Table~\ref{tab:naspi-mechanism} shows, the raw harmonic-signature method identifies the source correctly in all 36 settled windows (three windows per trajectory) for both mechanisms, confirming that $y_k$ in \eqref{eq:informationhierarchy} retains sufficient source-discriminating information. By contrast, the subsequent feature and decision stages, $\Phi$ and $\mathcal D$, may lose or misinterpret that information, yielding as few as 0/36 correct decisions for a mismatched feature--mechanism pairing. The governor--exciter contrast therefore demonstrates mechanism dependence of the tested features rather than a ranking among methods.

This distinction also explains why identical bus-level decisions need not imply equivalent feature maps. Under exciter forcing, DEF and CDEF both select bus 2604, yet the true bus 2634 improves from a median rank of 66.5 under DEF to 6 under CDEF, while its median score changes from $-7.24\times10^{-4}$ to $7.52\times10^{-3}$. The $K\operatorname{Re}\{\mathcal W_e\}$ term reverses the branch-score sign, but larger aggregated scores at bus 2604 still determine the same top-1 decision. Thus, CDEF materially changes the retained representation even when the final decision is unchanged.

\begin{table}[t]
\caption{Correct bus decisions at 2634/C in dataset Case 7}
\label{tab:naspi-mechanism}
\centering
\begin{tabular}{@{}lcc@{}}
\toprule
Decision rule & \shortstack{Governor reference} & \shortstack{Exciter reference} \\
\midrule
DEF with $K=0$ in in \eqref{eq:def-cdef-feature} & 36/36 & 0/36 \\
CDEF with $K=0.135$ in \eqref{eq:def-cdef-feature} & 36/36 & 0/36 \\
In-phase $Q$--$V$ measure in \eqref{eq:qvfeature}& 0/36 & 36/36 \\
\bottomrule
\end{tabular}
\par\smallskip
\parbox{\columnwidth}{\footnotesize Counts are over three windows each from 12 simulated trajectories.}
\end{table}

\noindent \emph{5) Case 3---measurement-induced ambiguity:}

The dataset on Case~3 uses matched SEXS voltage-reference interventions at buses 1131, the original source location in the contest, and 1032 at 0.379~Hz. Retaining only the fundamental coefficient of current magnitude on branch 1002--6506 gives one complex scalar $\alpha e^{j\omega t}$. Both signatures of 1131 and 1032 are nonzero; hence, their exact projective equivalence follows from Corollary~\ref{scalar-phasor}. Because this scalar record does not provide a complete voltage/current terminal pair, DEF, CDEF, and $Q$--$V$ measures are inapplicable. Eight separately simulated, amplitude/phase-aligned trajectory pairs give 24 paired-window comparisons, all within their empirical extraction tolerances. The raw locator returns $\widehat{\mathcal K}=\{1032,1131\}$ in all 48 individual-window decisions. These numerical results are consistent with the analytically established ambiguity, which arises at the measurement stage $y_k$ of \eqref{eq:informationhierarchy}. This example is intentionally constructed as a measurement-limited no-go case; it does not imply that the two buses remain indistinguishable with more PMU measurements.

Extending the measurements to 58 voltage/current measurement channels on the boundary branches of the South region that contains buses 1032 and 1131 separates the two candidate signatures and gives the correct source in all 48 decisions. The full terminal layout $M_{\rm T}$ further increases the separation, while adding generator measurements in $M_{\rm T+G}$ provides little additional improvement (Fig.~\ref{fig:naspi-structural}). Thus, unlike Cases~7 and 4, the limitation in Case~3 initially occurs at the measurement stage $y_k$ itself and is removed by adding measurement information.

\begin{figure}[t]
\centering
\includegraphics[width=0.55\columnwidth]{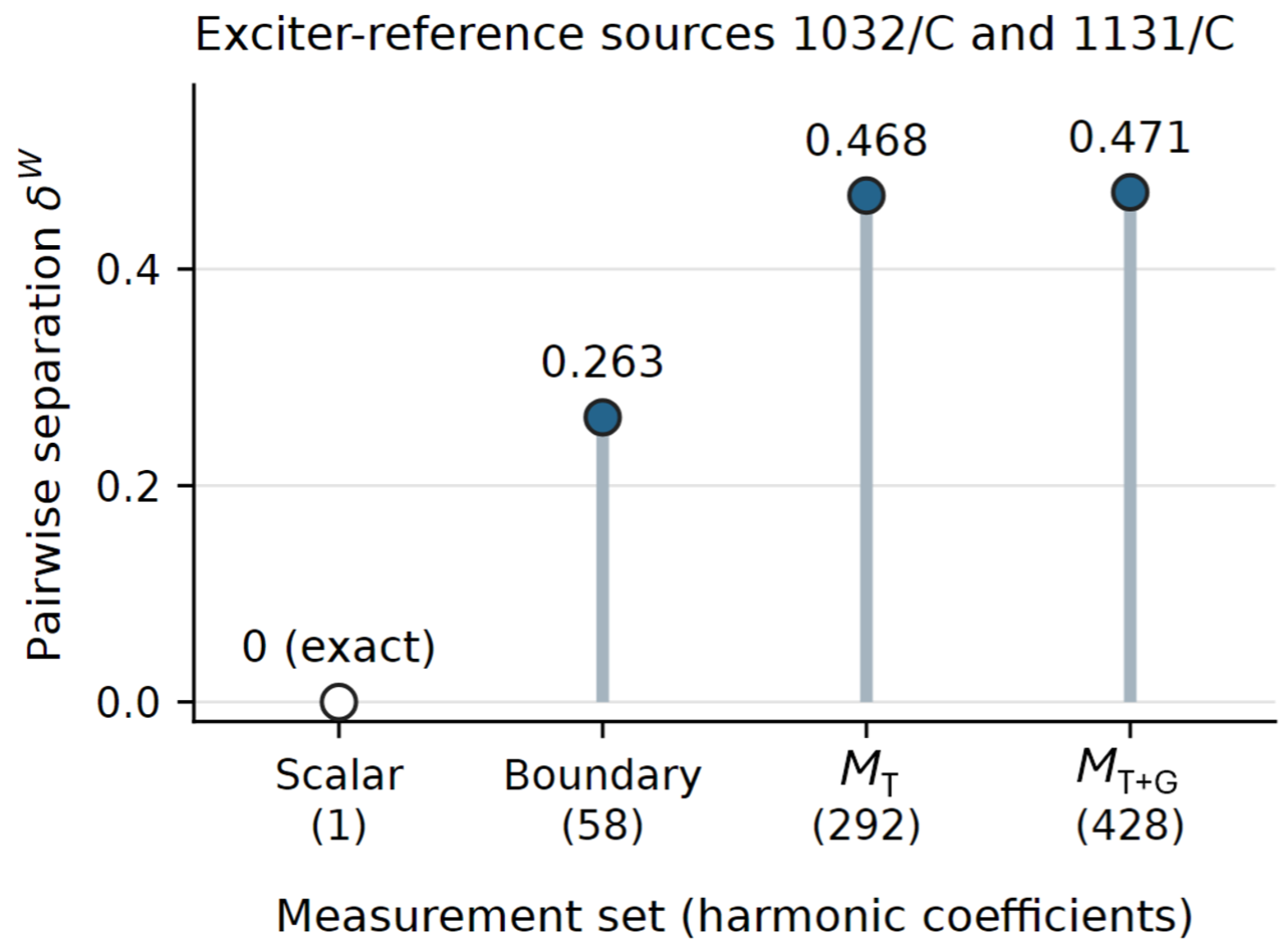}
\caption{Case~3: sensor augmentation for voltage-reference sources at buses 1032 and 1131. One scalar current-magnitude harmonic coefficient yields the two-bus ambiguity, which is recovered by three larger measurement layouts.}
\label{fig:naspi-structural}
\end{figure}

\vspace{-0.1in}
\section{Discussion and Conclusion}
\label{sec:discussion}
This paper addresses a conditional identifiability question: given a declared intervention dictionary, measurement set, operating-point model, forcing frequency, and single-source assumption, can the active intervention be uniquely identified from its steady-state harmonic response? Unknown forcing amplitude and phase reduce each nonzero candidate signature to its complex projective ray. Exact localization is possible if and only if these rays are distinct; collinear signatures are indistinguishable without additional information.

Structural identifiability does not ensure robust recovery, since closely spaced rays are sensitive to measurement and model errors. Proposition~3 gives a deterministic recovery condition for a fixed nominal dictionary, while parameter uncertainty may also perturb the signatures themselves. The framework separates raw identifiability, robustness, feature preservation, and mechanism-consistent interpretation.

The numerical studies demonstrate well-separated recovery, measurement-induced ambiguity and its removal, and sensitivity of nearly collinear signatures, while distinguishing terminal-feature performance, bus-level recovery, and internal-channel identifiability. They condition on the forcing frequency and are not end-to-end detection studies. The theory is limited to one constant-amplitude sinusoidal source and steady-state harmonic measurements; transient, time-varying, nonlinear, and multiple-source cases require separate formulations.

\vspace{-0.2in}
\section*{Acknowledgment}
OpenAI ChatGPT and Codex were used as research-assistance tools to assist with mathematical checking and automation of numerical studies. All results and conclusions were reviewed and verified by the author.

\vspace{-0.1in}
\bibliographystyle{IEEEtran}
\bibliography{references}

\end{document}